%% file: root.tex
\documentclass[letterpaper, 10 pt, conference]{ieeeconf}  % Comment this line out if you need a4paper

\IEEEoverridecommandlockouts                              % This command is only needed if 
\usepackage[T1]{fontenc}
\usepackage{amsmath,amsfonts,graphics,epsfig, amssymb,graphicx,color,dsfont}
 \usepackage{ragged2e}
\usepackage{lineno}
\usepackage{subcaption}
\newtheorem{thm}{Theorem}
\newtheorem{defi}{Definition}

\newtheorem{cor}{Corollary}

\newtheorem{example}{Example}

\input{macros}

\title{\LARGE \bf
Barrier Certificates for Assured Machine Teaching
}

\author{Mohamadreza Ahmadi$^{1}$, Bo Wu$^{2}$, Yuxin Chen$^{1}$, Yisong Yue$^{1}$, and Ufuk Topcu$^{2}$% <-this % stops a space
\thanks{This work was supported by AFOSR FA9550-19-1-0005, DARPA D19AP00004, NSF 1646522 and NSF 1652113.}% <-this % stops a space
\thanks{$^{2}$Institute for Computational Engineering and Sciences (ICES), University of Texas at Austin, Peter O'Donnel Jr. Building, 201 E 24th St. Austin, TX 78712, USA.
        {\tt\small \{bwu3,utopcu\}@utexas.edu}}%
\thanks{$^{1}$ California Institute of Technology 
1200 E. California Blvd., MC 104-44, Pasadena, CA 91125, USA.
        {\tt\small  \{mrahmadi,chenyux,yyue\}@caltech.edu }}%
}

\begin{document}

\maketitle
\thispagestyle{empty}
\pagestyle{empty}

%%%%%%%%%%%%%%%%%%%%%%%%%%%%%%%%%%%%%%%%%%%%%%%%%%%%%%%%%%%%%%%%%%%%%%%%%%%%%%%%
\begin{abstract}
Machine teaching can be viewed as optimal control for learning. Given a learner's model, machine teaching aims to determine the optimal training data to steer the learner towards a target hypothesis. In this paper, we are interested in providing assurances for machine teaching algorithms using control theory. In particular, we study a well-established learner's model in the machine teaching literature that is captured by the local preference over a version space.  We interpret the problem of teaching a preference-based learner as solving a partially observable Markov decision process (POMDP). We then show that the POMDP formulation can be cast as a special hybrid system, i.e., a discrete-time switched system. Subsequently, we use barrier certificates to verify set-theoric properties 
of this special hybrid system. We show how the computation of the barrier certificate can be decomposed and numerically implemented as the solution to a sum-of-squares (SOS) program. For illustration, we show how the proposed framework based on control theory can be used to verify the teaching performance of two well-known machine teaching methods. 
\end{abstract}

%%%%%%%%%%%%%%%%%%%%%%%%%%%%%%%%%%%%%%%%%%%%%%%%%%%%%%%%%%%%%%%%%%%%%%%%%%%%%%%%
\section{Introduction}

%In recent years, the need for machine learning systems has far exceeded the supply of machine learning experts.  Machine teaching, that is, algorithms designed to enable machines to teach other machines or humans, have thus received attention~~\cite{simard2017machine}. 
From an optimal control perspective, a learning system (e.g., a machine learning algorithm, or a human learner) defines a dynamical system where the state (i.e., learner's hypothesis) is driven by training data \cite{lessard2018optimal}. In this respect, machine teaching, i.e., the algorithmic framework of designing an optimal training set for learning a target hypothesis, can be viewed as optimal control for learning \cite{zhu2015machine}. In a typical setting of machine teaching, the target hypothesis is \emph{given} to the algorithm, and the goal of the teacher (machine) is to generate a minimal sequence of training examples such that the target hypothesis can be learned by a learner (human or another machine) from a finite set of hypotheses. 

One popular learner's model studied in the machine teaching literature is the version space learner. In such settings, the learner maintains a subset of hypotheses that are consistent with the examples received from a teacher, and outputs a hypothesis from this subset. %Depending on the learner's anticipated behavior, 
Based on different assumptions on the learner's behavior, multiple variants of the version space learner model has been studied in algorithmic machine teaching, leading to different notions of teaching complexity: For instance, (i) the ``worst-case'' model \cite{goldman1995complexity} assumes that the learner's behavior is completely unpredictable, and (ii) the ``preference-based'' model \cite{gao2017preference} assumes that she has a global preference over the hypotheses. These models are typically studied under the batch setting, where the teacher constructs a set of examples and provides them to the learner at once. Recently, \cite{chen18adaptive} considered the \emph{state-dependent} preference-based model, which generalizes the preference-based model of \cite{gao2017preference} to the adaptive setting. The state-dependent preference-based model assumes that the learner's choice of next hypothesis depends on some local preferences defined by the learner's state (i.e., the current hypothesis). 
%Due to its local preference structure, the state-dependent preference-based model is particularly suitable for designing and analyzing sequential teaching policies. Under the sequential setting,
In the sequential machine teaching setting, the teacher, after showing each example, obtains feedback about the hypothesis that the learner is currently entertaining; such feedback is further utilized to guide the selection of future teaching examples. %Such feedback can be used to select future teaching examples in a more informed way, and more importantly, to guide the teaching policy by integrating the ordering of training examples into the optimization process.

 In this paper, we use notions from hybrid systems analysis framework to study the state-dependent preference-based machine teaching model with the aim of verifying whether a given machine teaching method has assured teaching performance. We first show that state-dependent preference-based machine teaching model can be represented by a POMDP. Once this POMDP is formulated, we show that the evolution of the \emph{beliefs} over the states of this POMDP can be described by a discrete-time switched system (also see~\cite{pomdp1,ACJT}).
 %(also see our recent relevant works~\cite{pomdp1,ACJT}). 
 We use barrier certificates to verify whether the beliefs of this POMDP belong to some subset of the reachable belief space, which, in turn, corresponds to the probability of teaching of a hypothesis. From a computational standpoint, we show that these barrier certificates can be decomposed and constructed using SOS programming. We demonstrate the efficacy of our proposed methodology by comparing and analyzing two  machine teaching methods. 

The rest of this paper is organized as follows. We describe the state-dependent teaching model in the next section. In Section~\ref{sec:teachpomdp}, we propose a POMDP representation for machine teaching. In Section~\ref{sec:probform}, we briefly discuss a hybrid system that describe the evolution of this POMDP. In Section~\ref{sec:barrier}, we formulate a set of conditions based on barrier certificates for verifying the teaching performance and show how the calculations  can be decomposed. In Section~\ref{sec:SOS}, we propose a computational approach using SOS programming to find the barrier certificates. We elucidate the proposed method with an example in Section~\ref{sec:example} and conclude the paper in Section~\ref{sec:conclude}.

\textbf{Notation:} $\mathds{R}$ and $\mathds{N}$ denote the sets of real numbers and non-negative integers  $\{0,1,2,\ldots\}$, respectively.  $\mathds{N}_{\ge l}$, with $l \in \mathds{N}$, denotes $\{l,l+1,l+2,\ldots\}$. $\mathcal{R}[x]$ accounts for the set of polynomial functions with real coefficients in $z \in \mathds{R}^n$, $p: \mathds{R}^n \to \mathds{R}$ and $\Sigma \subset\mathcal{R}$ is the subset of polynomials with an SOS decomposition; i.e., $p \in  \Sigma[x]$ if and only if there are $p_i \in \mathcal{R}[x],~i \in \{1, \ldots ,k\}$ such that $p = p_i^2 + \cdots +p_k^2$.

\section{The State-Dependent Teaching Model}\label{sec:model}
We now state the adaptive machine teaching protocol, and describe the state-dependent learner's model of \cite{chen18adaptive}.

%%%%%%%%%%%%%%%%%%%%%%%%%%%%%%%%%%%%%%%%%%%%%%%
%%%%%%%%%%%%%%%%%%%%%%%%%%%%%%%%%%%%%%%%%%%%%%%
\subsubsection{The Teaching Domain}
Let $\Instances$ denote a ground set of unlabeled examples, and the set $\Clabels$ denotes the possible labels that could be assigned to elements of $\Instances$. We denote by $\Hypotheses$ a finite class of hypotheses, each element $\hypothesis\in \Hypotheses$ is a function $\hypothesis: \Instances \rightarrow \Clabels$. In our model, $\Instances$, $\Hypotheses$, and $\Clabels$ are known to both the teacher and the learner.
There is a target hypothesis $\hstar\in \Hypotheses$ that is known to the teacher, but not the learner.  Let $\Examples \subseteq \Instances \times \Clabels$ be the ground set of labeled examples. Each element $\example = (\instance_\example,\clabel^*_\example) \in \Examples$ represents a labeled example, where the label is given by the target hypothesis $\hstar$, i.e., $\clabel^*_\example = \hstar(\instance_\example)$.
Here, we define the notion of \emph{version space} needed to formalize our model of the learner. Given a set of labeled examples $\examples \subseteq \Examples$, the version space induced by $\examples$ is the subset of hypotheses $\Hypotheses(\examples) \in \Hypotheses$ that are consistent with labels of all the examples, i.e., $\Hypotheses(\examples):= \{\hypothesis: \hypothesis\in \Hypotheses  \texttt{\ and\ } \forall (\instance, \clabel) \in  \examples, h(\instance) = \clabel\}$.

%%%%%%%%%%%%%%%%%%%%%%%%%%%%%%%%%%%%%%%%%%%%%%%
%%%%%%%%%%%%%%%%%%%%%%%%%%%%%%%%%%%%%%%%%%%%%%%
\subsubsection{State-dependent Preference-based Model}\label{sec:model:learner}
The preference function encodes the learner's preferences of transitioning to a hypothesis.
Consider that the learner's current hypothesis is $h$, and there are two hypotheses $h'$, $h''$ that they could possibly pick as the next hypothesis. We define the preference function as $\sigma: \Hypotheses \times \Hypotheses \rightarrow \reals_+$. Given current hypothesis $h$ and any two hypothesis $\hypothesis', \hypothesis''$, we say that $\hypothesis'$ is preferred to $\hypothesis''$ from $\hypothesis$, iff $\sigma(\hypothesis' ; \hypothesis) < \sigma(\hypothesis''; \hypothesis)$. If $\sigma(\hypothesis'; \hypothesis) = \sigma(\hypothesis''; \hypothesis)$, then the learner could pick either one of these two.

The learner starts with an initial hypothesis $\hinit\in \Hypotheses$ before receiving any labeled examples from the teacher. Then, the interaction between the teacher and the learner proceeds in discrete time steps (trials). At any trial $t$, let us denote the labeled examples received by the learner up to (but not including) time step $t$ via a set $\examples^{t}$, the learner's version space as $\Hypotheses^t=\Hypotheses(\examples^{t})$,  and the current hypothesis as $\hypothesis^t$.
At trial $t$, we model the learning dynamics as follows:
\begin{enumerate}\denselist
\item the learner receives a new labeled example, %$\example^t$; %=(\instance^t, \clabel^t)
  and
\item the learner updates the version space $\Hypotheses^{t+1}$, and picks the next hypothesis based on the current hypothesis $\hypothesis^t$, version space $\Hypotheses^{t+1}$, and the preference function $\ordering$:
%\item \yuxin{We should add that we are considering the adaptive setting---the learning protocol we consider in the experiments assumes that the teacher, at each iteration, also observes the learner's current hypothesis. This will cause the issue that the setting now reduces to MDPs instead of POMDPs. Perhaps we should add the following: ``(Optionally), the teacher observes the learner's next hypothesis $\hypothesis^{t+1}$''.}
\end{enumerate}
 \begin{align}
    \hypothesis^{t+1} \in \{\hypothesis\in \Hypotheses^{t+1} : \orderingof{\hypothesis}{\hypothesis^t} = \min_{\hypothesis'\in \Hypotheses^{t+1}} \orderingof{\hypothesis'}{\hypothesis^t}\}.
    \label{eq.learners-jump}
  \end{align}

\subsubsection{The Teaching Protocol and Objective}
The teacher's goal is to steer the learner towards the target hypothesis $\hstar$ by providing a sequence of labeled examples. At trial $t$, we consider the following teaching protocol:
\begin{enumerate}
    \item the teacher selects an unlabeled example $\instance^t \in \Instances$ and presents it to the learner;
    \item the learner makes a guess of the label, %of $\instance^t$, denoted by
    i.e. $\clabel^t := h^t(\instance^t)$.
    \item the teacher receives feedback from the learner\footnote{We consider two variants of the learner feedback: (a) the teacher indirectly observes the learner's hypothesis $h^t$ via label $\clabel^t$; (b) the teacher directly observes the learner's current hypothesis $h^t$. Our analysis in the subsequent sections applies to both scenarios. For discussion simplicity we focus on the more general setting (a) in Section III-VI.}  %assume the more general setting (a). In fact, one can view (b) as a special case of (a) when }
    % guess $h_t(\clabel^t)$, 
    and provides the true label $\hstar(\instance^t)$;
    \item the learner transitions from the current $\hypothesis^t$ to the next hypothesis $\hypothesis^{t+1}$ as per the model described in the previous subsection.
    \item Teaching finishes if the learner's updated hypothesis $\hypothesis^{t+1}=\hypothesis^*$. 
\end{enumerate}
The goal of teaching algorithms is to achieve this goal in the minimal number of time steps.

% \begin{itemize}
%     \item the teacher selects a labeled example $\example^t \in \Examples$ and the learner transitions from the current $\hypothesis^t$ to the next hypothesis $\hypothesis^{t+1}$ as per the model described above.
%     \item Teaching finishes here if the learner's updated hypothesis $\hypothesis^{t+1}=\hypothesis^*$. 
% \end{itemize}
%\yuxin{I added the following paragraph---we may need to shorten it.}
%Note the state-dependent teaching model described above generalizes many existing models of version-space learners in the algorithmic machine teaching literature. For example, if the preference function is constant, i.e., $\forall \hypothesis,\hypothesis'\in \hypotheses$, $\sigma(\hypothesis' ; \hypothesis) \equiv \text{const}$, then it reduces to the worst-case model of \cite{goldman1995complexity}; if the preference function does not depend on the current hypothesis i.e., $\forall \hypothesis, \hypothesis',\hypothesis''\in \hypotheses$, $\sigma(\hypothesis' ; \hypothesis) = \sigma(\hypothesis' ; \hypothesis'')$, then it reduces to the preference-based model of \cite{gao2017preference}. 

The state-dependent teaching model is also found to be consistent with % with proper choices of the preference function $\sigma$, such model is consistent with
simple human learning models in cognitive science, including the ``win-stay lose-shift'' model \cite{bonawitz2014win,rafferty2016faster} (e.g., when $\sigma(\hypothesis'; \hypothesis) = 0$ if $\hypothesis= \hypothesis'$ and~$1$ otherwise, the learner prefers to stay at the same hypothesis if it is consistent with the observed data).

% {\color{blue}  YUXIN: A REMARK ON WHY THE PREFERENCE FUNCTION MODEL MAKES SENSE? ANY HEURISTICS?}

\section{POMDP Model for Machine Teaching}\label{sec:teachpomdp}

Given the state-dependent teaching model as described in Section~\ref{sec:model}, we can represent machine teaching as a sequential decision making under uncertainty scenario. To this end, we propose a POMDP representation for the learner based on the state-dependent teaching model. The POMDP model can be described as follows.

%\yuxin{This setting doesn't fit exactly into the version space model we consider in the reference. The version space model is actually an MDP instead of a POMDP, i.e., we assume that the teacher observes the learner's state $h_{t}$ at each time step $t$. In such cases, $p_0$ is a delta distribution, }
\begin{defi}[Learning POMDP]
  The  learning POMDP $\mathcal{P}_L$ is a tuple %$(\mathcal{H}, p_0, \mathcal{Z}, T, \mathcal{Y}, O)$
  $(\mathcal{H}, p_0, \mathcal{Z}, T, \Observations, O)$
  
  %\yuxin{$\mathcal{Y}$ is fully known to the teacher, because the teacher knows $\hstar$. The way that the current POMDP is described looks like the model of an (active) learning system, rather than a machine teaching model. The current model is fine for the non-adaptive teaching setting. We will have to change the description $b_t$ to be the belief of the teacher (as I commented below in the next paragraph). For the adaptive setting, the teacher's belief on the learner's hypothesis also depends on the learner's feedback---I think one possible fix is to enrich the observation set $\mathcal{Y}$, such that it also captures the feedback that the teacher receives from the learner.}, where
\begin{itemize}
\item the hypotheses set $\mathcal{H}$ is a finite set of hidden states;
\item $p_0$ is the probability of having an initial hypothesis $h_0 \in \mathcal{H}$;
\item the set of labelled examples $\mathcal{Z}$  constitute the finite set of actions;
\item $T$ describes the transitions from one hypothesis (state) to another characterized by the preference functions  as given by~\eqref{eq.learners-jump};
\item $\Observations$ denotes the set of observations made by the teacher.
\item $O(y_t \mid h_t,z_t)$ is determined by the current hypothesis function.
\end{itemize}
\end{defi}

Here, the observation model $O(y_t \mid h_t, z_t)$ defines how the version space gets updated. When referring to the ``version space'' learners, we are implicitly considering the ``noise-free'' setting, i.e., all consistent hypotheses are uniformly distributed, or equivalently, $O(y_t \mid h_t, z_t)$ is binary. Moreover, according to~(\ref{eq.learners-jump}), the transition function $T(h, z_{t-1}, h')$ defines a \emph{uniform} distribution: the learner only goes to the hypotheses $h'$ that are the most preferred; hence, $T$ induces a uniform distribution over the most preferred hypothesis according to the preference function $\sigma$.

%\yuxin{The description below fits the active learning setting, but not teaching. I think it is important that the algorithm takes the perspective of the teacher instead of the learner}
The learner starts with an initial hypothesis $h_0$ and over a sequence of trials, in which an example $z_t\in \mathcal{Z}$ is shown and the learner receives a corresponding observation $y_t \in \mathcal{Y}$, develops a belief in the new hypothesis $h$.
% \yuxin{see comments in the next sentence}. In other words, the learner updates its belief in the hypothesis $b_t(h)$ based on the examples $z\in \mathcal{Z}$ and the observations $y\in \mathcal{Y}$ it receives over time. 
%\yuxin{perhaps it makes more sense if we say ``the teacher develops a belief of which state (i.e., hypothesis) the learner is in''; and ``the teacher updates its belief $b_t(h)$ in the learner's current hypothesis $h$, based on the examples...''}
Then, the hypothesis belief evolves according to
\begin{multline}\label{eq:beliefstPOMDP}
b_t(h')=\\ \frac{O(y_{t}\mid h',z_{t-1})\sum_{h\in \mathcal{H}}T(h,z_{t-1},h')b_{t-1}(h)}{\sum_{h'\in \mathcal{H}}O(y_{t} \mid h',z_{t-1})\sum_{h\in \mathcal{H}}T(h,z_{t-1},h')b_{t-1}(h)},
\end{multline}
%\yuxin{in the numerator, subscript of the sum: should $q\in Q$ be $h\in \mathcal{H}$?}

The objective of a teaching policy is then to assure that the learner learns the target hypothesis $h^* \in \mathcal{H}$ in $t^*$ number of trials. That is, 
\begin{equation}\label{eq:teachperform}
b_{t^*}(h^*) \ge \lambda,
\end{equation}
where we refer to $0<\lambda\le1$ as the \emph{teaching performance}. In addition, given a teaching policy, we are often interested in finding the minimum number of trials such that the learner learns a target hypothesis, i.e., 
\begin{eqnarray}\label{eq:teachperform2}
\min~t^*~\text{subject to}~b_{t^*}(h^*) \ge \lambda.
\end{eqnarray}

Ideally, given a pre-specified number of trials ${t^*}$, a teaching algorithm is \emph{perfect},  if $\lambda = 1$, i.e., the probability of learning the target hypothesis after ${t^*}$ number of examples is one. However, achieving a perfect teaching algorithm in $t^{*}$ number of trials may not be realistic. In practice, it is desirable that we teach the target hypothesis with teaching performance $\lambda \ge 0.75$.

%%%%%%%%%%%%%%%%%%%%%%%%%%%%%%%%%%%%%%%%%%%%%%%%%%%%%%%%%%%%%%%%%%%%%%%%%%%%%%%%

%We denote by $\mathcal{C}^m(X)$, with $z \subseteq \mathds{R}^n$, the space of $m$-times continuously differentiable functions and by  $\partial^m = \frac{\partial^m}{\partial x^m}$ the derivatives up to order $m$. 

%\section{Preliminaries}

%\subsection{Partially Observable Markov Decision Processes}

\section{Belief Evolution as a Hybrid System}\label{sec:probform}

Checking whether~\eqref{eq:teachperform} holds by solving the learning POMDP directly is a PSPACE-hard problem~\cite{ChatterjeeCT16}.  In this section, we show that the learning POMDP can be represented as a special hybrid system~\cite{goebel2009hybrid}, specifically, a discrete-time switched system~\cite{ahmadi2008non,KUNDU2017191,zhang2009exponential}. 

The belief update equation~\eqref{eq:beliefstPOMDP}  can be characterized as a discrete-time switched system, where the actions $a \in A$ define the switching modes.  Formally, the hypothesis belief \emph{dynamics}~\eqref{eq:beliefstPOMDP} can be described as
\begin{equation}\label{eq:beliefstPOMDPdyn}
b_t = f_z\left(b_{t-1},y_t\right),
\end{equation}
where $b$ denote the belief vector belonging to the belief unit simplex $\mathcal{B}$ and $b_0 =p_0$. In~\eqref{eq:beliefstPOMDPdyn}, $z \in \mathcal{Z}$ denote the examples that can be interpreted as the indices for the switching modes, $y \in \mathcal{Y}$ are the observations representing inputs, and $t \in \mathds{N}_{\ge 1}$ denote the discrete time instances. The (rational) vector fields $\{f_{z}\}_{z \in \mathcal{Z}}$ with $f_z: [0,1]^{|\mathcal{Z}|} \times \mathcal{Y} \to [0,1]^{|\mathcal{Z}|} $ are described as the vectors with rows
$$
f_z^{h'}(b,y) = \frac{O(y \mid h',z)\sum_{h\in \mathcal{H}}T(h,z,h')b_{t-1}(h)}{\sum_{h'\in \mathcal{H}}O(y \mid h',z)\sum_{h\in \mathcal{H}}T(h,z,h')b_{t-1}(h)},
$$
where $f_z^{h'}$ denotes the $h'$th row of $f_z$.  

%If the transition probabilities are uncertain, i.e., they belong to some given set, the system can be represented as an uncertain discrete-time switched system
%\begin{equation}\label{equation:discretesystem}
%b_t = f_a\left(b_{t-1},\theta,z_t\right),
%\end{equation}
%where $\theta \in \Theta$ is a set of uncertain parameters and $\Theta$ represents the uncertain transition probability intervals \eqref{eq:uncertaintpdf}. That is, 
%\begin{equation*}
%		\theta_{q,a,q'} = T(q,a,q') \in [\underline{l}_{q,a,q'},\overline{l}_{q,a,q'}],~ (q,a,q') \in T_u,
%\end{equation*}
%and
%\begin{equation}\label{eq:uncertainparams}
%\Theta =\left \{ \theta \mid \theta_{q,a,q'}  \in  [\underline{l}_{q,a,q'},\overline{l}_{q,a,q'}],~ (q,a,q') \in T_u \right \}.
%\end{equation}

We consider two classes of problems in learning POMDP verification:
\begin{itemize}
\item [1.] \textit{Arbitrary-Policy Verification}: This case corresponds to analyzing~\eqref{eq:beliefstPOMDPdyn} under \emph{arbitrary switching} with switching modes determined by the examples $z \in \mathcal{Z}$.
\item [2.] \textit{Fixed-Policy Verification}: This corresponds to analyzing~\eqref{eq:beliefstPOMDPdyn} under \emph{state-dependent switching}. In fact, a teaching policy $\pi:\mathcal{B} \to \mathcal{Z}$ (a mapping from the hypothesis beliefs into examples) determines regions in the belief space where each mode (example) is active.
\end{itemize}

Both cases of switched systems with~\emph{arbitrary switching} and~\emph{state-dependent switching} are well-known in the systems and controls literature~(see~\cite{liberzon2003switching,hespanha2004uniform} and references therein). 

%\begin{example}Consider a POMDP with two states $\{q_1,q_2\}$, two actions $\{a_1,a_2\}$, and $z \in {Z}$. The policy
%\begin{equation} \label{eq:example-sds}
%\pi=
%\left\{\begin{array}{lr}
%        a_1, & b \in \mathcal{B}_1,\\
%        a_2, & b \in \mathcal{B}_2\\
%        \end{array}\right.
%\end{equation}
%leads to two different switching modes based on whether the states belong to the regions $\mathcal{B}_1$ or $\mathcal{B}_2$ (see Figure~\ref{figure1}). That is, the belief update equation~\eqref{equation:discretesystem1} is given by 
%\begin{equation} \label{eq:example-dynamics}
%b_t=
%\left\{\begin{array}{lr}
%        f_{a_1}\left(b_{t-1},z_t\right), & b \in \mathcal{B}_1,\\
%        f_{a_2}\left(b_{t-1},z_t\right), & b \in \mathcal{B}_2.\\
%        \end{array}\right.
%\end{equation}
%Note that the belief space is given by $\mathcal{B}=\mathcal{B}_1 \cup \mathcal{B}_2 = \{ b \mid b(q_1)+b(q_2)=1\}$.
%\end{example}
%\begin{figure}[tbp] 
%\begin{center} 
%\includegraphics[width=5cm]{Allerton_POMDP_Figure1}
%\vspace{-1cm}
%\caption{An example of a POMDP with two states and the state-dependent switching modes induced by the policy~\eqref{eq:example-sds}. }
%\label{figure1}
%\end{center}
%\end{figure}

\section{Verifying Teaching Performance Using Barrier Certificates}\label{sec:barrier}

In the following, we describe a method based on barrier certificates to verify the teaching performance as given by~\eqref{eq:teachperform}. We then focus on the two cases of arbitrary policy verification and fixed-policy verification. We further show that in both cases, the calculation of the barrier certificates can be decomposed. 

In order to check the teaching performance, we consider following teaching-failure set
\begin{equation}\label{eq:unsafeset}
\mathcal{B}_f = \{ b \in \mathcal{B} \mid b_{t^*}(h^*) < \lambda \},
\end{equation}
which is the complement of~\eqref{eq:teachperform}.

We  have the following result.

\begin{thm}\label{theorem-barrier-discrete}
Given the learning POMDP~ $(\mathcal{H}, p_0, \mathcal{Z}, T, \mathcal{Y}, O)$, a target hypothesis $h^* \in \mathcal{H}$, and a teaching performance $\lambda$, and a pre-set number of trials $t^*$, if there exists a function $B:\mathds{N} \times \mathcal{B} \to \mathds{R}$ called the barrier certificate such that
\begin{equation}\label{equation:barrier-condition1}
B(t^*,b_{t^*})  > 0, \quad \forall b_{t^*} \in \mathcal{B}_f,
\end{equation}
with $\mathcal{B}_f$ as described in~\eqref{eq:unsafeset},
\begin{equation}\label{equation:barrier-condition11}
 B(0,b_0) < 0, \quad \text{for} \quad b_0=p_0,
\end{equation}
and
\begin{multline}\label{equation:barrier-condition2}
B\left(t,f_z(b_{t-1},y)\right) - B(t-1,b_{t-1}) \le 0, \\ \forall t \in \{1,2,\ldots,t^*\},~ \forall z \in \mathcal{Z},~\forall y \in \mathcal{Y}, ~\forall b \in \mathcal{B},
\end{multline}
then there the teaching performance $\lambda$ is satisfied, i.e., inequality~\eqref{eq:teachperform} holds.
\end{thm}
\begin{proof}
The proof is carried out by contradiction. Assume at trial ${t^*}$, the teaching performance is not satisfied. Thus, there is a solution to the hypothesis belief update equation~\eqref{eq:beliefstPOMDPdyn}  with $b_0=p_0$ such that  $b_{t^*}(h^*) < \lambda$. From inequality~\eqref{equation:barrier-condition2}, we have 
$$
B(t,b_t) \le B(t-1,b_{t-1})
$$
for all $t\in \{1,2,\ldots,{t^*}\}$ and all examples $z \in \mathcal{Z}$. Hence, $B(t,b_t) \le B(0,b_0)$ for all $t \in \{1,2,\ldots,{t^*}\}$. Furthermore, inequality~\eqref{equation:barrier-condition11} implies that 
$$
B(0,b_0) < 0 
$$
for $b_0 =p_0$.  Since the choice of ${t^*}$ can be arbitrary, this is a contradiction because it implies that $B({t^*},b_{t^*}) \le B(0,b_0) < 0$. Therefore, there exist no solution of \eqref{eq:beliefstPOMDPdyn} such that $b_0=p_0$ and $b_{t^*} \in \mathcal{B}_f$ for any sequence of examples $z \in \mathcal{Z}$. Hence, the teaching performance  is satisfied.
\end{proof}

In practice, we may have a large number of examples. Then, finding a barrier certificate that satisfies the conditions of Theorem~\ref{theorem-barrier-discrete} becomes prohibitive to compute. In the next result, we show how the calculation of the barrier certificate can be decomposed into finding a set of barrier certificates for each example and then taking the convex hull of them.

\begin{thm}\label{theorem-barrier-convexhull}
Given the learning POMDP~ $(\mathcal{H}, p_0, \mathcal{Z}, T, \mathcal{Y}, O)$, a target hypothesis $h^* \in \mathcal{H}$, a teaching performance $\lambda$, and a pre-set number of trials $t^*$, if there exists a set of function $B_z:\mathds{N} \times \mathcal{B} \to \mathds{R}$, $z \in \mathcal{Z}$, such that
\begin{equation}\label{equation:barrier-condition1x}
B_z(t^*,b_{t^*})  > 0, \quad \forall b_{t^*} \in \mathcal{B}_f,~~\forall z \in \mathcal{Z},
\end{equation}
with $\mathcal{B}_f$ as described in~\eqref{eq:unsafeset},
\begin{equation}\label{equation:barrier-condition11x}
 B_z(0,b_0) < 0, \quad \text{for} \quad b_0=p_0,~~\forall z \in \mathcal{Z},
\end{equation}
and
\begin{multline}\label{equation:barrier-condition2x}
B_z\left(t,f_z(b_{t-1},y)\right) - B_z(t-1,b_{t-1}) \le 0, \\ \forall t \in \{1,2,\ldots,t^*\},~ \forall z \in \mathcal{Z},~\forall y \in \mathcal{Y}, ~\forall b \in \mathcal{B},
\end{multline}
then there the teaching performance $\lambda$ is satisfied, i.e., inequality~\eqref{eq:teachperform} holds. Furthermore, the overall barrier certificate is given by $B = \text{co}\{B_z\}_{x \in \mathcal{Z}}$.
\end{thm}
\begin{proof}
The proof was omitted due to lack of space here. Please refer to the extended version~\cite{ahmadi2018barrier}.
\end{proof}

The efficacy of the above result is that we can search for each example-based barrier certificate $B_z$, $z \in \mathcal{Z}$, independently or in parallel and then verify whether the overall teaching algorithm (described by the learning POMDP) satisfies a pre-specified teaching performance (see Fig.~\ref{figure2} for an illustration). 

Next, we demonstrate that, if a teaching policy is given, the search for the barrier certificate can be decomposed into the search for a set of local barrier certificates. As discussed earlier, a teaching policy  $\pi:\mathcal{B} \to \mathcal{Z}$ assigns an example to different regions of the belief space (refer to Section IV). Without loss of generality, we consider policies of the form
\begin{equation}\label{eq:policy}
\pi(b) = \begin{cases} z_1, & b \in \mathcal{B}_1, \\ \vdots & \vdots \\   z_{|\mathcal{Z}|}, & b \in \mathcal{B}_N,    \end{cases}
\end{equation}
where $N$ denotes the number of partitions of $\mathcal{B}$ and $\cup_{i=1}^N \mathcal{B}_i = \mathcal{B}$. Note that the number of partitions and the number of examples are not necessarily equal. We denote by $z_i$ the example active in the partition $\mathcal{B}_i$. 

\begin{thm}\label{theorem-barrier-policygiven}
Given the learning POMDP~ $(\mathcal{H}, p_0, \mathcal{Z}, T, \mathcal{Y}, O)$, a target hypothesis $h^* \in \mathcal{H}$, a teaching performance $\lambda$, a teaching policy $\pi:\mathcal{B} \to \mathcal{Z}$ as described in~\eqref{eq:policy}, and a pre-set number of trials $t^*$, if there exists a set of function $B_i:\mathds{N} \times \mathcal{B}_i \to \mathds{R}$, $i \in \{1,2,\ldots,N\}$, such that
\begin{equation}\label{equation:barrier-condition1xx}
B_i(t^*,b_{t^*})  > 0, \quad \forall b_{t^*} \in \mathcal{B}_f \cap \mathcal{B}_i ,~~i \in \{1,2,\ldots,N\},
\end{equation}
with $\mathcal{B}_f$ as described in~\eqref{eq:unsafeset},
\begin{equation}\label{equation:barrier-condition11xx}
 B_i(0,b_0) < 0, \quad \text{for} \quad b_0=p_0,~~i \in \{1,2,\ldots,N\},
\end{equation}
and
\begin{multline}\label{equation:barrier-condition2xx}
B_i\left(t,f_{z_i}(b_{t-1},y)\right) - B_i(t-1,b_{t-1}) \le 0, \\ \forall t \in \{1,2,\ldots,t^*\},~\forall y \in \mathcal{Y}, ~\forall b \in \mathcal{B}_i,\\~i \in \{1,2,\ldots,N\},
\end{multline}
then there the teaching performance $\lambda$ is satisfied, i.e., inequality~\eqref{eq:teachperform} holds. Furthermore, the overall barrier certificate is given by $B = \text{co}\{B_i\}_{i=1}^N$.
\end{thm}
\begin{proof}
The proof was omitted due to lack of space here. Please refer to the extended version~\cite{ahmadi2018barrier}.
\end{proof}

\begin{figure}[tbp] 
\begin{center} 
\includegraphics[width=7cm]{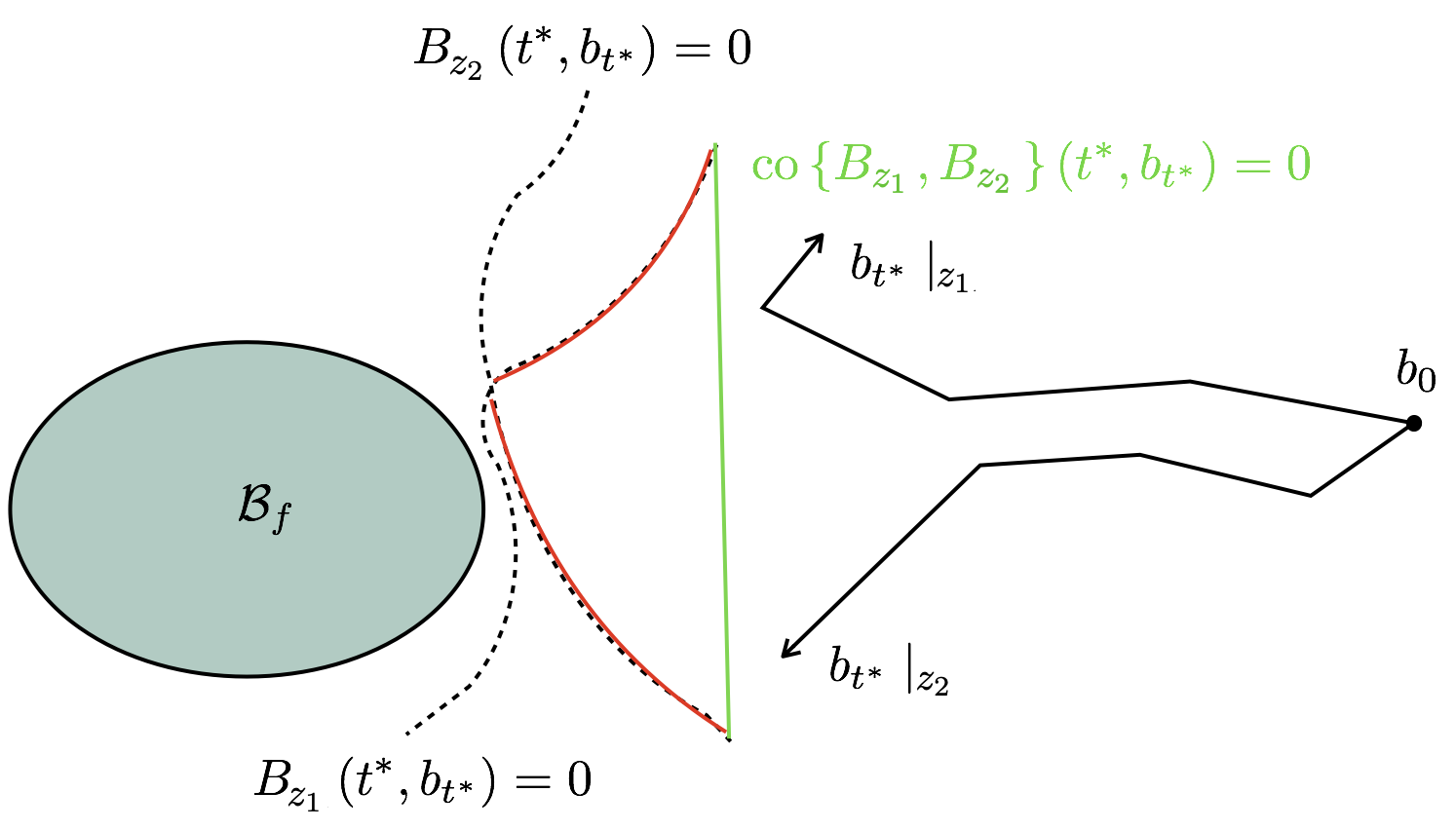}
\caption{Decomposing the barrier certificate computation for a learning POMDP with two examples $z_1$ and $z_2$: the zero-level sets of $B_{z_1}$ and $B_{z_2}$ at trial $t^*$ separate the evolutions of the hypothesis beliefs starting at $b_0$ from $\mathcal{B}_f$. The green line illustrate the zero-level set of the barrier certificate formed by taking the convex hull of $B_{z_1}$ and $B_{z_2}$.}
\label{figure2}
\end{center}
\end{figure}
%The technique used in Corollary~\ref{corollary-barrier-discrete:optimality} is analogous to the one used in~\cite{7171125,AHMADI201733} for bounding (time-averaged) functional outputs of systems described by partial differential equations. The method proposed here, however, can be used for a large class of discrete time systems and the belief update equation is a special case that is of our interest.
%
%In practice, it is often desirable to make sure a design is both optimal and safe. The problem can be described by checking whether the solutions of the belief update switched dynamics~\eqref{equation:discretesystem} enter the following set 
%$$
%\mathcal{B}_u = \mathcal{B}_u^s \cup \mathcal{B}_u^o.
%$$
%To this end, we can adopt either of the following approaches (see Figure~\ref{figure2}). 

We proposed two techniques for decomposing the construction of the barrier certificates and checking a pre-set teaching performance. Our method relied on barrier certificates that take the form of the convex hull of a set of local barrier certificates~(see similar results in~\cite{8264626,2018arXiv180104072A}). Though the convex hull barrier certificate may introduce a level of conservatism, it is computationally easier to find (as will be discussed in more detail in the next section).  We remark that another technique that can be used for  decomposition may use non-smooth barrier certificates~\cite{7937882}, i.e., max or min of a set of local barrier certificates.

%The following result can be inferred from Theorem~\ref{} in~\cite{}, which states that we can compute the barrier certificates for safety and optimality separately (or in parallel) and then form a non-smooth barrier certificate composed of the two...]\

\section{Computational Method via SOS Programming} \label{sec:SOS}

In this section, we propose techniques for finding the barrier certificates and checking whether a teaching performance is satisfied using SOS programming~\cite{Par00,PAVPSP13}. 

In order to cast the conditions of Theorem~1-3 into SOS programs, we need polynomial/rational variables and require the associated sets to be semi-algebraic. Fortunately, these requirements naturally fit our problem. The hypothesis belief space is a semi-algebraic set. Moreover, the right-hand side of the belief update equation~\eqref{eq:beliefstPOMDPdyn} is composed of rational functions in the belief states~$b_t(h)$,~$h \in \mathcal{H}$. That is, 
\begin{multline}\label{eq:belief-update-rational}
b_t(h') = \frac{S_z\left( b_{t-1}(h'),y_{t-1} \right)}{R_z\left( b_{t-1}(h'),y_{t-1}  \right)} \\
= \frac{O(h',z_{t-1},y_{t})\sum_{h\in \mathcal{H}}T(h,z_{t-1},h')b_{t-1}(h)}{\sum_{h'\in \mathcal{H}}O(h',z_{t-1},y_{t})\sum_{h\in \mathcal{H}}T(h,z_{t-1},h')b_{t-1}(h)}.
\end{multline}
Furthermore, the teaching-failure  set~\eqref{eq:unsafeset} is a semi-algebraic set. 

At this point, we  present conditions based on SOS programs to verify a given teaching performance of a teaching algorithm. 

\begin{cor}\label{cor:SOS-Safety}
Given the learning POMDP~ $(\mathcal{H}, p_0, \mathcal{Z}, T, \mathcal{Y}, O)$, a target hypothesis $h^* \in \mathcal{H}$, a teaching performance $\lambda$, and a pre-set number of trials $t^*$, if there  exist polynomial functions $B \in \mathcal{R}[t,b]$ of degree $d$ and   $p^f \in {\Sigma}[b]$, and constants $s_1,s_2>0$ such that
\begin{equation}\label{eq:setssos1}
B\left({t^*},b_{t^*}\right) +   p^f(b_{t^*}) \left( b_{t^*}(h^*) -  \lambda   \right)- s_1 \in \Sigma \left[b_{t^*}\right],
\end{equation}
\begin{equation}\label{eq:setssos2}
-B\left(0,p_0\right)  - s_2 >0,
\end{equation}
and 
\begin{multline}\label{eq:setssos3}
- {R_z\left( b_{t-1} \right)}^d\bigg(B\left(t,\frac{S_z\left( b_{t-1},y \right)}{R_z\left( b_{t-1},y \right)} \right) - B(t-1,b_{t-1})   \\  \in \Sigma[t,b_{t-1}],  \forall  t \in \{1,2,\ldots,{t^*}\},~y \in \mathcal{Y},~z \in \mathcal{Z},
\end{multline}
then there exists no solution of~\eqref{eq:beliefstPOMDPdyn} such that $b_0 =p_0$ and $b_{t^*} \in \mathcal{B}_f$ and, hence, the teaching performance is satisfied.
\end{cor}
\begin{proof}
The proof was omitted due to lack of space here. Please refer to the extended version~\cite{ahmadi2018barrier}.
\end{proof}

Similarly, we can formulate SOS feasibility conditions for checking the inequalities in Theorem~2.

\begin{cor}\label{cor:SOS-Safety2}
Given the learning POMDP~ $(\mathcal{H}, p_0, \mathcal{Z}, T, \mathcal{Y}, O)$, a target hypothesis $h^* \in \mathcal{H}$, a teaching performance $\lambda$, and a pre-set number of trials $t^*$, if there  exist polynomial functions $B_z \in \mathcal{R}[t,b]$, $z \in \mathcal{Z}$, of degree $d$ and   $p^f_z \in {\Sigma}[b]$, $z \in \mathcal{Z}$, and constants $s^1_z,s^2_z>0$, $z \in \mathcal{Z}$, such that
\begin{multline}\label{eq:setssos12}
B_z\left({t^*},b_{t^*}\right) +   p^f_z(b_{t^*}) \left( b_{t^*}(h^*) -  \lambda   \right) \\
- s^1_z \in \Sigma \left[b_{t^*}\right],~~z \in \mathcal{Z},
\end{multline}
\begin{equation}\label{eq:setssos22}
-B_z\left(0,p_0\right)  - s_z^2 >0,~~z \in \mathcal{Z},
\end{equation}
and 
\begin{multline}\label{eq:setssos32}
- {R_z\left( b_{t-1} \right)}^d\bigg(B_{x}\left(t,\frac{S_z\left( b_{t-1},y \right)}{R_z\left( b_{t-1},y \right)} \right) - B_{x}(t-1,b_{t-1}) \bigg)   \\  \in \Sigma[t,b_{t-1}],  \forall  t \in \{1,2,\ldots,{t^*}\},\\~y \in \mathcal{Y},~z \in \mathcal{Z},
\end{multline}
then there exists no solution of~\eqref{eq:beliefstPOMDPdyn} such that $b_0 =p_0$ and $b_{t^*} \in \mathcal{B}_f$ and, hence, the teaching performance is satisfied.
\end{cor}

We assume that a teaching policy in the form of~\eqref{eq:policy} assigns examples to semi-algebraic partitions of  the hypothesis belief space $\mathcal{B}$ described as
\begin{equation}
\mathcal{B}_i =\left\{  b \in \mathcal{B} \mid g_i(b) \le 0      \right\},~~i \in \{1,2,\ldots,N\}.
\end{equation}
We then have the following SOS formulation for Theorem~3 using Positivstellensatz.

\begin{cor}\label{cor:SOS-Safety3}
Given the learning POMDP~ $(\mathcal{H}, p_0, \mathcal{Z}, T, \mathcal{Y}, O)$, a target hypothesis $h^* \in \mathcal{H}$, a teaching performance $\lambda$, a teaching policy $\pi:\mathcal{B} \to \mathcal{Z}$ as described in~\eqref{eq:policy}, a teaching performance $\lambda$, and a pre-set number of trials $t^*$, if there  exist polynomial functions $B_i \in \mathcal{R}[t,b]$, $i \in \{1,2,\ldots,N\}$, of degree $d$, $p^{l_1}_i \in {\Sigma}[b]$, $i \in \{1,2,\ldots,N\}$, $p^{l_2}_i \in {\Sigma}[b]$, $i \in \{1,2,\ldots,N\}$, $p^{l_3}_i \in {\Sigma}[b]$, $i \in \{1,2,\ldots,N\}$, and   $p^f_i \in {\Sigma}[b]$, $i \in \{1,2,\ldots,N\}$, and constants $s^1_i,s^2_i>0$, $i \in \{1,2,\ldots,N\}$, such that
\begin{multline}\label{eq:setssos12}
B_i\left({t^*},b_{t^*}\right) +   p^f_i(b_{t^*}) \left( b_{t^*}(h^*) - \lambda   \right) +p^{l_1}_i(b_{t^*})g_i(b_{t^*}) \\
- s^1_i \in \Sigma \left[b_{t^*}\right],~~i \in \{1,2,\ldots,N\},
\end{multline}
\begin{equation}\label{eq:setssos22}
-B_i\left(0,p_0\right) +p^{l_2}_i(p_0)g_i(p_0)  - s_i^2 >0,~~i \in \{1,2,\ldots,N\},
\end{equation}
and 
\begin{multline}\label{eq:setssos32}
- {R_z\left( b_{t-1} \right)}^d\bigg(B_{i}\left(t,\frac{S_z\left( b_{t-1},y \right)}{R_z\left( b_{t-1},y \right)} \right) - B_{i}(t-1,b_{t-1}) \bigg)   \\  +p^{l_3}_i(b_{t-1})g_i(b_{t-1})  \in \Sigma[t,b_{t-1}],  \forall  t \in \{1,2,\ldots,{t^*}\},\\~y \in \mathcal{Y},~z \in \mathcal{Z},~~i \in \{1,2,\ldots,N\},
\end{multline}
then there exists no solution of~\eqref{eq:beliefstPOMDPdyn} such that $b_0 =p_0$ and $b_{t^*} \in \mathcal{B}_f$ and, hence, the teaching performance is satisfied.
\end{cor}

\begin{figure*}[!t]
  \centering
  \begin{subfigure}[b]{0.88\textwidth}
    \includegraphics[trim={8pt 5pt 8pt 5pt},height=.11\textwidth]{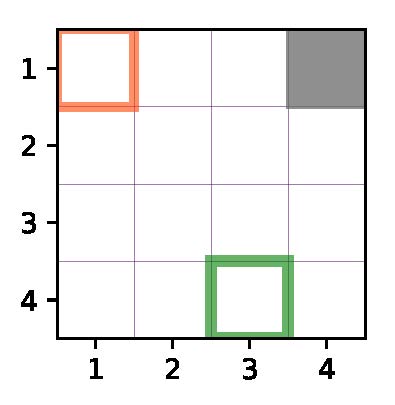}
    \includegraphics[trim={8pt 5pt 8pt 5pt},height=.11\textwidth]{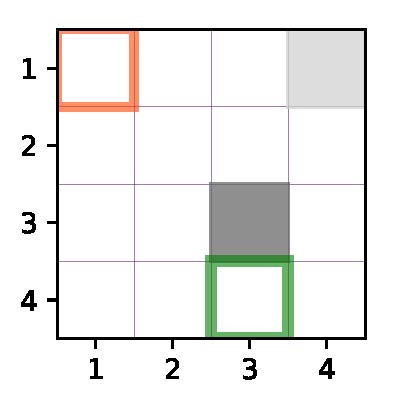}
    \includegraphics[trim={8pt 5pt 8pt 5pt},height=.11\textwidth]{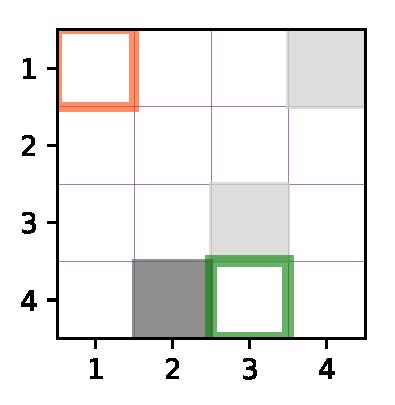}
    \includegraphics[trim={8pt 5pt 8pt 5pt},height=.11\textwidth]{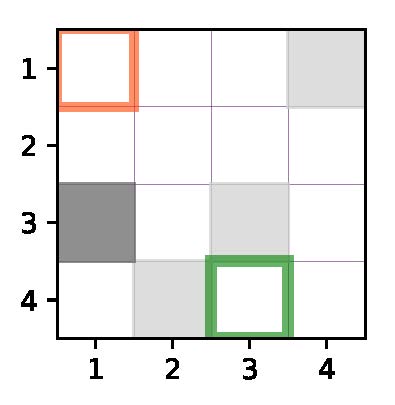}
    \includegraphics[trim={8pt 5pt 8pt 5pt},height=.11\textwidth]{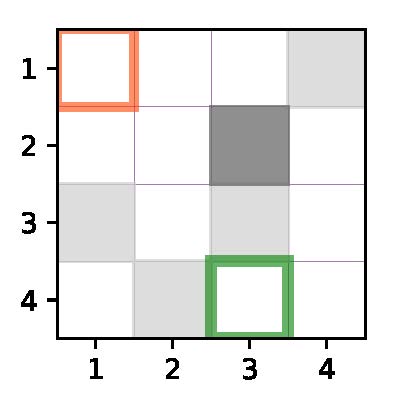}
    \includegraphics[trim={8pt 5pt 8pt 5pt},height=.11\textwidth]{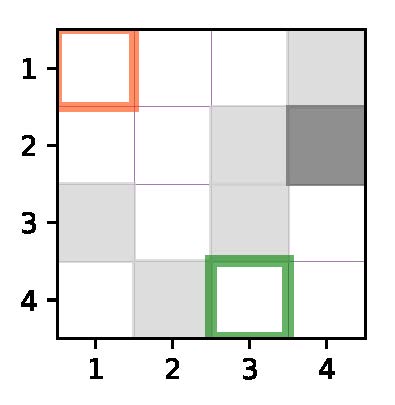}
    \includegraphics[trim={8pt 5pt 8pt 5pt},height=.11\textwidth]{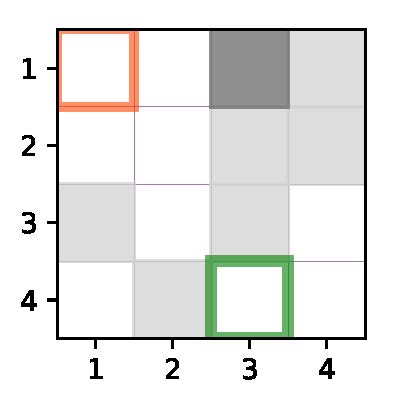}
    \includegraphics[trim={8pt 5pt 8pt 5pt},height=.11\textwidth]{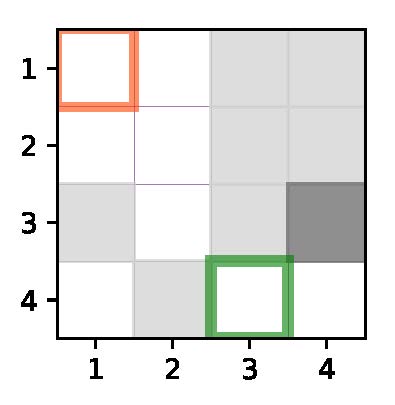}
    \includegraphics[trim={8pt 5pt 8pt 5pt},height=.11\textwidth]{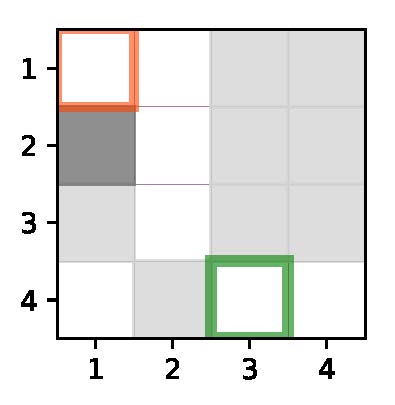}
    \caption{Myopic}
    \label{fig:app:lattice:myopic}
  \end{subfigure}
  \begin{subfigure}[b]{0.88\textwidth}
    \includegraphics[trim={8pt 5pt 8pt 5pt},height=.11\textwidth]{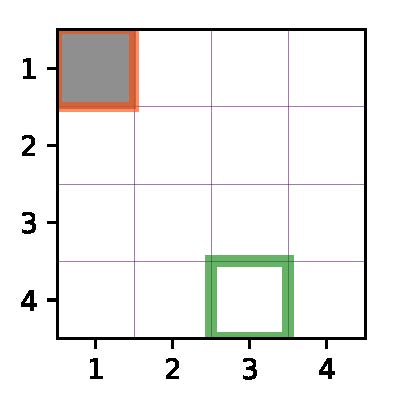}
    \includegraphics[trim={8pt 5pt 8pt 5pt},height=.11\textwidth]{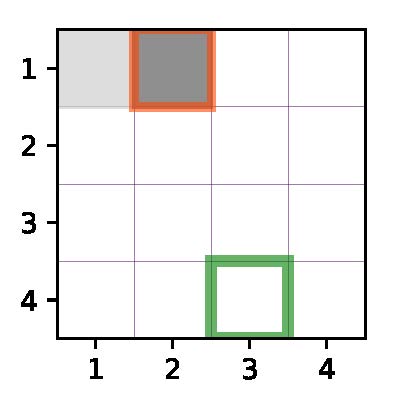}
    \includegraphics[trim={8pt 5pt 8pt 5pt},height=.11\textwidth]{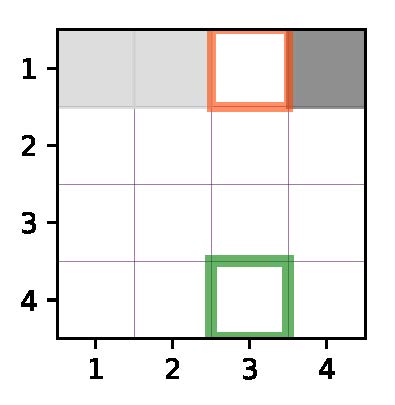}
    \includegraphics[trim={8pt 5pt 8pt 5pt},height=.11\textwidth]{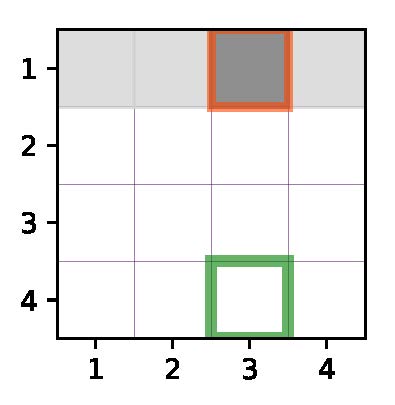}
    \includegraphics[trim={8pt 5pt 8pt 5pt},height=.11\textwidth]{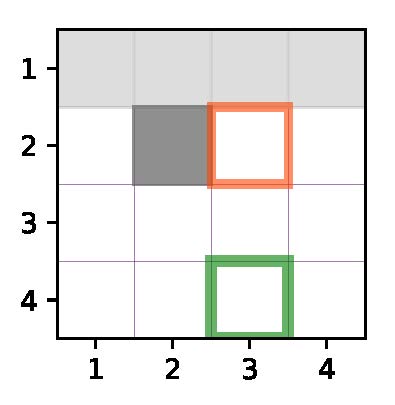}
    \includegraphics[trim={8pt 5pt 8pt 5pt},height=.11\textwidth]{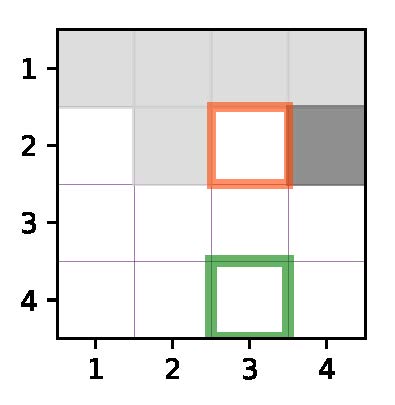}
    \includegraphics[trim={8pt 5pt 8pt 5pt},height=.11\textwidth]{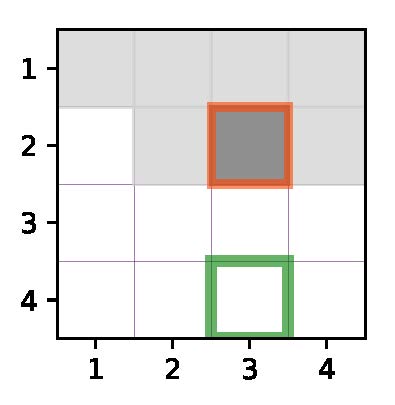}
    \includegraphics[trim={8pt 5pt 8pt 5pt},height=.11\textwidth]{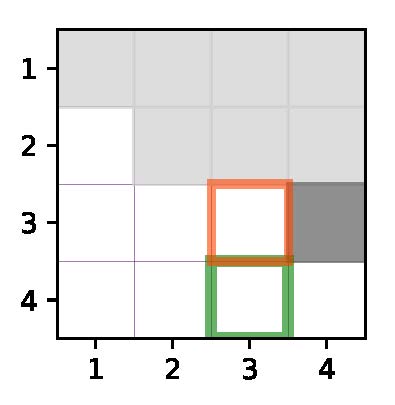}
    \includegraphics[trim={8pt 5pt 8pt 5pt},height=.11\textwidth]{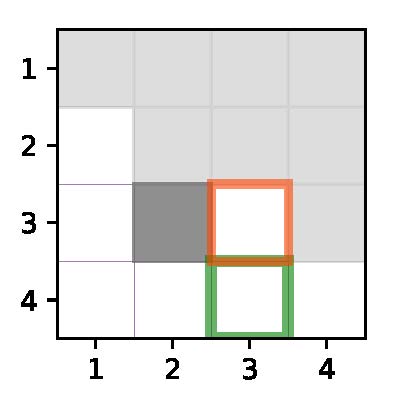}
        \caption{Ada-L}
    \label{fig:app:lattice:adal}
  \end{subfigure}
  \caption{Teaching sequences generated by  Myopic and Ada-L algorithms on a $4\times 4$ lattice, with $h_0=(1,1), h^*=(3,4)$. The learner's initial hypothesis is marked by orange, and the target is marked by green. The dark gray square represents the teaching example at the current time step, while light gray squares represent the previous teaching examples.}
  \label{fig:app:lattice_3}
\end{figure*}

\section{Example}\label{sec:example}
%\yuxin{add discussion about protocol change}
In order to illustrate the proposed framework, we consider a toy scenario, where the teacher aims to teach/steer a human
learner to reach a goal state in a physical environment. Each hypothesis/node corresponds to some
unexplored territory, and there exists an example which flags the territory as explored. The learner
prefers local moves, and if all neighboring territories are explored, the learner jumps to the next
closest one.

The physical environment is characterized by a $4\times 4$ lattice corresponding to $16$ hypotheses. The target hypothesis is located at $h^*=(4,4)$. The teacher has $16$ choices of locations on the lattice to show to the student as examples. The student then receives two labels based on its answer $y \in \{-1,1\}$.  The preference function $\sigma(h';h)$ is given by the minimum distance between hypotheses described by~$\ell_1(h';h)$.

In this example, we compare two teaching algorithms in the adaptive setting, where the teacher observes the learner's hypothesis at each iteration. The Myopic algorithm is a greedy approach which, at each iteration, picks the teaching example such that after observing the label, the worst-case rank of the target hypothesis in the learner's resulting version space is the smallest. The Ada-L algorithm aims to teach the learner some intermediate hypothesis at each iteration, i.e., it aims to direct the learner to transit to a hypothesis that is ``closer'' to the target hypothesis. For more details of the algorithms please refer to~\cite{chen18adaptive}.

Each algorithm provides a set of policies for which we seek to find the minimum number of trials such that the following teaching performance is assured
$$
b_{t^*} (h^*) \ge \lambda.
$$
To this end, we minimize the number of trials $t^*$ such that (27)-(29) are satisfied. We start by a large number of trials ($16$ in this case) and decrease it until no barrier certificate can be found to verify the teaching performance.  We fix the degree of variables $B_i$, $p^{l_1}_i$, $p^{l_2}_i$, $p^{l_3}_i$, and   $p^f_i \in {\Sigma}[b]$, $i \in \{1,2,\ldots,N\}$ in Corollary~3 to~$2$ and search for the certificates. In order to check the SOS conditions formulated in Section~\ref{sec:SOS}, we use diagonally-dominant-SOS (DSOS) relaxations of the SOS programs implemented through the SPOTless tool~\cite{spot} (for more details see~\cite{ahmadi2017dsos,8263706}).

The results on finding the minimum number of trials $t^*$ for which the teaching performance is satisfied were as follows. 
% \subsubsection{$h_0=(1,1)$ and $h^*=(4,4)$}  For the Myopic algorithm, we could not find any certificate for $\lambda=0.8$. Changing the the teaching performance to $\lambda=0.6$ yielded certificates for only $t^*=14$. In contrast, for the Ada-L algorithm, we could obtain $t^*=9$ assuring teaching performance $\lambda=0.8$ and $t^*=10$ guaranteeing teaching performance $\lambda=0.9$.

\subsubsection{$h_0=(1,1)$ and $h^*=(3,4)$}  For the Myopic algorithm, we could not find any certificate for $\lambda=0.8$. Changing the the teaching performance to $\lambda=0.55$ yielded certificates for only $t^*=15$. On the other hand, for the Ada-L algorithm, we obtained $t^*=9$ assuring teaching performance $\lambda=0.8$ and $t^*=10$ assuring teaching performance $\lambda=0.9$.

%For example, in the first column of \figref{fig:app:lattice_1}, Myopic (Algorithm 1) picks any teaching example, because the worst-case rank of the target hypothesis at $(4,4)$ is the same for any teaching example. The Ada-L algorithm, on the other hand, tries to direct the learner to make a local move from $(1,1)$ to $(1,2)$ or $(2,1)$, which are closer to the target at $(4,4)$. 
%\yuxin{}

The results can also be corroborated from simulations. As can be see in Figure~\ref{fig:app:lattice_3}, the Myopic algorithm  perform poorly on simple teaching tasks as compared to the Ada-L algorithm.

\section{CONCLUSIONS}\label{sec:conclude}

We presented a method based on barrier certificates to assure the performance of machine teaching algorithms.  Our computational method was in terms of SOS programs, where we used DSOS relaxations. It was shown in~\cite{zheng2018sparse} that using sparse SOS (SSOS) programs leads to more efficient and less conservative results. Future work can explore the use of more scalable SOS relaxations such as SSOS.

\bibliographystyle{IEEEtran}
\bibliography{reference}

\end{document}

%% file: macros.tex
\newcommand{\Observations}{\mathcal{Y}}

\newcommand{\reals}{\ensuremath{\mathbb{R}}}

\usepackage{bbm}

\newcommand{\denselist}{\itemsep 0pt\topsep-6pt\partopsep-6pt}
\newcommand{\hypothesis}[0]{\ensuremath{h}}
\newcommand{\Hypotheses}{\mathcal{H}}

\newcommand{\hstar}{\hypothesis^*}
\newcommand{\hinit}{{\hypothesis^0}}

\newcommand{\Examples}{\mathcal{Z}}
\newcommand{\Instances}{\mathcal{X}}
\newcommand{\Clabels}{\mathcal{Y}}

\newcommand{\examples}{Z}

\renewcommand{\example}{{z}}
\newcommand{\instance}{{x}}
\newcommand{\clabel}{{y}}

\newcommand{\ordering}{\sigma}

\newcommand{\orderingof}[2]{\ordering({#1};{#2})}